\documentclass[conference, a4paper]{IEEEtran}
\IEEEoverridecommandlockouts
\usepackage{cite}
\usepackage{amsmath,amssymb,amsfonts}
\usepackage{graphicx}
\usepackage[font=footnotesize]{caption}
\usepackage{textcomp}
\usepackage{xcolor}
\usepackage{url}
\usepackage{booktabs}
\usepackage{geometry}
\newtheorem{proposition}{Proposition}
\begin{document}

\title{Learning to Contest: Decentralized Robust Fairness\\
in Cooperative MARL via Cross-Attention}

\makeatletter
\def\@IEEEauthorblockNfont{\fontsize{11pt}{13pt}\selectfont}
\def\@IEEEauthorblockAfont{\fontsize{10pt}{12pt}\selectfont}
\makeatother

\author{\IEEEauthorblockN{Can Savc\i}
\IEEEauthorblockA{\textit{Independent Researcher}\\
\.{I}zmir, Turkey \\
savcisavcican@gmail.com}}

\maketitle

\begin{abstract}
Fair cooperative multi-agent reinforcement learning (MARL) teams that maximize an
egalitarian welfare are exploitable: a single self-interested agent free-rides on
the surplus that fair agents forgo to raise the worst-off, and the known remedy
is a \emph{centralized} need-based allocator. We show that a decentralized
defense becomes possible once contention is \emph{graded}: when a contested
resource still delivers a fraction $1{-}c$, a worst-off cooperator that contests
a free-rider strictly improves on yielding (Prop.~1), so leverage exists for
every $c<1$. We introduce \textbf{CAN}, a permutation-equivariant
cross-attention policy over agents' observed behaviour that infers how many
free-riders are present and responds proportionally---turn-taking when none,
contesting just enough when some. Trained against an adversarial league, CAN
keeps best-response exploitability near the centralized oracle
($\rho{\approx}1.2$--$1.5$ vs.\ $\rho{=}N$ unprotected) at essentially no
efficiency cost, where the fair-MARL learners (GGF, FEN, SOTO) each collapse to
an exploitable or wasteful extreme. Giving those objectives CAN's identical
adversarial training does not rescue them, so the objective---not adversarial
training alone---is what makes hardening possible. Against a \emph{committed}
(non-adaptive) defector, every learned defense including ours provides
deterrence rather than immunity, weakening as the leverage $(1{-}c)/2$
vanishes. Across further environments and team sizes the same principle sets
the scope: robustness holds exactly as far as the game's contest leverage
reaches, and we map that boundary rather than claim to remove it.
\end{abstract}

\renewcommand{\IEEEkeywordsname}{Keywords}

\section{Introduction}
Cooperative multi-agent reinforcement learning (MARL) increasingly optimizes not
only efficiency but \emph{fairness}: methods such as FEN~\cite{jiang2019fen} and
SOTO~\cite{zimmer2021soto} maximize a fairness-inducing welfare (typically the
Generalized Gini Welfare~\cite{siddique2020ggf}) over agents' returns so that no
agent is starved. These methods assume every agent cooperates in being fair. When
one does not---one stakeholder upgrades a controller to minimize its own delay, one
client maximizes its own throughput---the fair team is exploitable: because
team-oriented agents \emph{sacrifice} their own utility to raise the worst-off, a
self-interested agent free-rides on that sacrifice~\cite{ivanov2023mediated}. In a
strictly rivalrous (all-or-nothing) resource---where a contested unit is either won
by one agent or wasted---this exploit is hard to defend against \emph{at the policy
level}: a cooperator that contests a free-rider merely causes a collision, gaining
nothing, so a welfare-fair team is indifferent between yielding and contesting; a
recent study formalizes this futility and the centralized need-based remedy that
sidesteps it~\cite{savci2026exploit}. The robust remedy is then to take allocation
out of the agents' hands entirely, via a \emph{centralized} need-based mechanism.
This paper takes up the harder question
that centralization sidesteps: \emph{can decentralized policies---acting only on
local, observed information---retain fairness under a self-interested agent?}

We answer affirmatively, with caveats, in three steps.

\textbf{(1) Leverage exists once contention is graded.} The policy-level futility
above hinges on \emph{all-or-nothing} contention: a contested resource is either
won by one agent or wasted, so a cooperator that contests a free-rider gains
nothing. We introduce a \emph{graded-contention} game in which $m\ge 2$ claimers
split a fraction $1{-}c$ of the resource (waste $c$). We prove (Prop.~1) that for
any $c<1$ a worst-off cooperator that contests a lone free-rider receives
$(1{-}c)/2>0$ instead of $0$, so contesting \emph{strictly dominates} yielding:
decentralized leverage reappears. All-or-nothing is the limiting case $c{=}1$.

\textbf{(2) The residual difficulty is coordination under uncertainty.} Leverage is
not a policy. The number of free-riders $D$ is unknown and varies episode to
episode (including $D{=}0$). A fixed rule cannot win: always-contest pays the
waste $c$ even when nobody defects, and always-yield collapses the moment someone
does. A robust cooperator must \emph{infer} $D$ from observed behaviour and respond
\emph{proportionally}---do nothing (turn-take) when $D{=}0$, contest just enough
when $D\ge 1$.

\textbf{(3) A cross-attention policy, trained against a league, realizes it.} We
introduce \textbf{CAN} (Cross-Attention Networks), a shared,
permutation-equivariant policy in which each agent attends over the \emph{observed
behaviour} of all agents (utilities, claim-rates, who is worst-off) to estimate
contention and act. What is load-bearing is \emph{behaviour-conditioning with
adaptive aggregation}---the response must depend on \emph{how many} others are
grabbing, not on fixed identities---and among permutation-respecting aggregators
cross-attention is the most stable instance (Sec.~\ref{sec:ablation}). We train
CAN against an adversarial league (PSRO~\cite{lanctot2017psro}) and audit it with
both a fresh best-response defector and a committed always-claim script. CAN keeps
trained-best-response exploitability low ($\rho{\approx}1.2$--$1.5$) across all
contention levels while wasting essentially nothing when no free-rider is present
(efficiency ${\approx}1.0$), approaching the centralized oracle without any central
allocator; against the \emph{committed} defector its robustness is a deterrence
effect that decays as leverage vanishes (Sec.~\ref{sec:hardened}).

We are deliberate about limits. We report a negative-leaning result---single
adversary co-training is unreliable at low contention---and two genuine
fragilities: zero-shot transfer of an $N{=}6$ policy to teams of $12$ and $24$
degrades at high contention, and robustness to a \emph{committed} (rather than
adaptive) defector decays as leverage vanishes, for every method including ours
(Sec.~\ref{sec:hardened}). The contribution is a controlled,
honest map of how far decentralized robust fairness reaches, not a claim that the
decentralization gap is closed.

\section{Background and Related Work}
\textbf{Welfare-based fair MARL.} Given per-agent returns
$\mathbf{u}=(u_1,\dots,u_N)$, fair-MARL maximizes a Schur-concave welfare
$W(\mathbf{u})$ that is increasing (efficiency) and favours equity, e.g.\ the
Generalized Gini Welfare $W(\mathbf{u})=\sum_k w_k u^{\uparrow}_k$ with decreasing
weights $w_k$ so the worst-off agent dominates~\cite{siddique2020ggf}.
FEN~\cite{jiang2019fen} and SOTO~\cite{zimmer2021soto} are representative learners.
These objectives create the appropriable surplus we must defend.

\textbf{Exploitability and the decentralization gap.} The vulnerability is a
fairness-specific instance of free-riding in sequential social
dilemmas~\cite{leibo2017ssd}; that welfare-maximizing cooperation is gameable in
general was noted by Ivanov et al.~\cite{ivanov2023mediated}, who restore
cooperation through an incentive-compatible mediator. Under strictly rivalrous
(all-or-nothing) contention, policy-level reciprocity has no efficient lever, and a
centralized need-based allocator---which decides who receives each resource---is an
upper bound on fairness that any decentralized method should aim to
match~\cite{savci2026exploit}. We take that target seriously and pursue the
decentralized side it leaves open.

\textbf{Reciprocity and conditional cooperation.} Inequity
aversion~\cite{hughes2018inequity} and learned
reciprocity~\cite{eccles2019reciprocity} make cooperation conditional on others'
behaviour in social dilemmas. Our cooperators are likewise behaviour-conditioned,
but the obstacle here is sharper: under all-or-nothing contention reciprocity has
no efficient lever at all; we show graded contention restores one and that a
learned attention policy can exploit it.

\textbf{Attention in MARL.} Attention over agents is well established for
value/critic mixing and communication, e.g.\ actor-attention
critics~\cite{iqbal2019maac} and the centralized critics of
MADDPG~\cite{lowe2017maddpg}. We use a single-head attention
block~\cite{vaswani2017attention} \emph{in the decentralized policy}, each agent's
query attending across all agents' observed-behaviour tokens (Sec.~IV); its role is
contention inference, and its
permutation-equivariance gives a single $N$-agnostic policy that can be evaluated at
unseen team sizes.

\textbf{League / population training.} Training against a growing population of
past best responses---fictitious play~\cite{heinrich2015fsp} and
PSRO~\cite{lanctot2017psro}---yields policies robust to a history of adversaries
rather than to a single co-evolving one. We use it to stabilize the cooperator
against the worst free-rider it has seen, which we find is what makes robustness
consistent across contention levels.

\section{The Graded-Contention Game}
\textbf{Game.} $N$ agents act over $T$ steps. Each step every agent chooses
\textsc{Claim} ($1$, the self-interested act) or \textsc{Yield} ($0$). Let
$\mathbf{u}\ge 0$ accumulate the resource each agent has collected and let $m$ be
the number of claimers this step. The unit resource is allocated by the
\emph{graded-contention} rule
\begin{equation}
g(m)=\begin{cases}
1 & m=1,\\[2pt]
(1{-}c)\,/\,m & m\ge 2,\\[2pt]
\text{routed to the worst-off agent} & m=0,
\end{cases}
\end{equation}
where each claimer receives $g(m)$ and $c\in[0,1]$ is the contention
\emph{waste}. A sole claimer wins the whole unit; $m\ge 2$ claimers split a
discounted $1{-}c$; if nobody claims, the resource flows (losslessly) to the
worst-off agent. The parameter $c$ interpolates from lossless sharing ($c{=}0$) to
the all-or-nothing model ($c{=}1$, where contesting destroys the resource). One
caveat is worth stating plainly: the $m{=}0$ branch is itself a need-based
mechanism---when nobody claims, the \emph{rule}, not an agent, routes the unit to
the worst-off. Our no-central-allocator claim therefore concerns contested steps:
whenever $m\ge1$, who gets what is decided entirely by the agents' own
\textsc{Claim}/\textsc{Yield} choices, and that is where the robustness problem
lives.

\textbf{Threat model.} Each episode, each agent is independently a \emph{defector}
with some probability, so the number of free-riders $D$ is variable and unknown
(including $D{=}0$); defectors always claim. The remaining \emph{cooperators} share
one policy and seek to keep the allocation fair and efficient. We audit a trained
cooperator policy by \emph{freezing} it and training a fresh best-response defector
to maximize its own utility---an exploitability test stronger than a fixed
always-claim adversary.

\textbf{Bounded free-ride metric.} We measure exploitation by the free-ride factor
of the defector \emph{group}: its share of the delivered total, normalized by its
fair share $n_{\mathrm{def}}/N$,
\begin{equation}
\rho \;=\; \frac{N\,\sum_{i\in\mathcal{D}} u_i}
{\,n_{\mathrm{def}}\,\sum_{j} u_j\,}\;\in\;[0,\,N/n_{\mathrm{def}}],
\label{eq:rho}
\end{equation}
where $\mathcal D$ is the defector set. $\rho{=}1$ is exactly fair; $\rho{=}N$ means
a lone defector took everything. Eq.~\eqref{eq:rho} is \emph{bounded} by
construction, unlike the naive ratio of defector utility to mean-cooperator
utility, which diverges when cooperators are starved as $c\to1$; the delivered
total is always positive, so $\rho$ stays in range and is comparable across $c$. We
also report Jain's index~\cite{jain1984fairness} and efficiency
(delivered$/T$).

\textbf{Leverage.} Under all-or-nothing contention ($c{=}1$) a cooperator cannot
deny a free-rider except by wasteful levelling-down---contesting collides and
destroys the unit---so a welfare-fair team is indifferent between yielding and
contesting and has no efficient defense. Graded contention breaks that premise.

\begin{proposition}[Decentralized leverage under graded contention]
Fix $c<1$ and a step at which a single free-rider claims. Let a worst-off
cooperator choose \textsc{Claim} rather than \textsc{Yield}. Then the cooperator's
received resource rises from $0$ to $(1{-}c)/2>0$, while the free-rider's falls
from $1$ to $(1{-}c)/2$. Hence contesting \emph{strictly dominates} yielding for the
worst-off cooperator, and strictly reduces the free-rider's appropriated surplus.
\end{proposition}

\noindent\emph{Proof.} If the worst-off cooperator yields, the free-rider is the
sole claimer ($m{=}1$), receives $g(1){=}1$, and the cooperator receives $0$. If
the cooperator contests, $m{=}2$ and each receives $g(2){=}(1{-}c)/2$. Since $c<1$,
$(1{-}c)/2>0$, so the cooperator strictly gains and the free-rider strictly loses.
$\square$

\noindent The dominant action is also \emph{objective-aligned}: lifting the
worst-off cooperator from $0$ to $(1{-}c)/2$ raises the cooperator mean and lowers
their spread, so it strictly increases the training welfare
$W_{\mathrm{coop}}{=}\mathrm{mean}{-}\mathrm{std}$ (Sec.~IV). Contesting is thus
rewarded both at the level of the individual worst-off agent and of the team
objective.

\noindent\emph{Dominance, not equilibrium.} Proposition~1 is a per-step dominance
statement, not an equilibrium analysis of the repeated game: it fixes a step on
which a single free-rider claims and gives the worst-off cooperator's myopic best
reply. It does not assert that always-contesting is an equilibrium of the $T$-step
game (it is not: when $D{=}0$, contesting burns $c$ for nothing), nor that a
defector's best response is to always claim---against a team that contests
proportionally, each extra claim earns only the marginal share $(1{-}c)/m$.
Empirically the two come apart: the \emph{trained} best-response defector is
deterrable---against a league-trained cooperator it converges to \emph{end-game}
claiming (claim rate ${\approx}0$ over the first three episode quarters,
${\approx}0.65$ in the last, at both $c{=}0.5$ and $c{=}0.9$), harvesting where
the finite horizon leaves punishment no future to bite---whereas a
\emph{committed} always-claim script ignores punishment, and at high $c$ it
extracts more than the trained best response finds (Sec.~\ref{sec:hardened}). Our audits therefore report the worse of the two. A full
equilibrium characterization of the repeated graded-contention game is open; our
claims are the existence of the per-step lever and the empirical exploitability
audit built on it.

\noindent Proposition~1 establishes that the lever \emph{exists} for every $c<1$;
it does not say how to pull it. The cooperator must contest \emph{only} when a
free-rider is present (else it pays the waste $c$ for nothing) and only \emph{enough}
cooperators must contest (over-contesting wastes; under-contesting leaves surplus).
Because $D$ is unknown and variable, this is an inference-and-coordination problem,
which motivates a behaviour-conditioned learned policy.

\section{Method: Cross-Attention Networks (CAN)}
\textbf{Observed-behaviour tokens.} At step $t$ each agent $i$ is represented by a
token of six features computable from public state---no access to others'
intentions or rewards:
\begin{equation}
\textstyle
x_i=\Big[\tfrac{u_i}{T},\,\tfrac{u_i-\bar u}{T},\,\tfrac{u_i-u_{\min}}{T},\,
\mathbf 1[u_i{=}u_{\min}],\,\tfrac{c\!c_i}{t},\,\tfrac{t}{T}\Big],
\end{equation}
where $\bar u,u_{\min}$ are the team mean/min utilities and $cc_i/t$ is agent $i$'s
running claim-rate. The claim-rate channel is what lets a cooperator estimate
\emph{how many} others are grabbing.

\textbf{Architecture.} CAN's aggregator is a single-head attention block over the
$N$ tokens: with $Q,K,V$ linear projections of $\{x_i\}$,
$\mathrm{ctx}=\mathrm{softmax}(QK^\top/\sqrt{d})V$; each agent concatenates its
token with its context, passes a $\tanh$ layer, and outputs
\textsc{Claim}/\textsc{Yield} logits (Fig.~\ref{fig:arch}a). We use
\emph{cross-attention} in the inter-agent sense---agent $i$'s query attends
\emph{across} all agents' behaviour tokens to build its context---computed for all
agents at once as a single self-attention pass over the token matrix; it is not the
two-sequence encoder--decoder block sometimes meant by the same term. The block is shared by
all cooperators and \emph{permutation-equivariant}, so (i) the policy is
decentralized---agent $i$ acts on the public tokens, and on every contested step
($m\ge1$) no central allocator decides the allocation (the $m{=}0$ fallback is the
mechanism's, Sec.~III)---and (ii) it is $N$-agnostic, enabling zero-shot evaluation
at unseen team sizes.

\textbf{Cooperator objective.} Cooperators maximize a per-step welfare reward-to-go
with welfare
\begin{equation}
W_{\mathrm{coop}}=\mathrm{mean}_{i\in\mathcal C}(u_i)-\mathrm{std}_{i\in\mathcal
C}(u_i),
\end{equation}
over cooperators $\mathcal C$: the mean term rewards reclaiming utility from a
free-rider (contest when $D\ge1$), the std term rewards turn-taking (share evenly),
and doing nothing when $D{=}0$ avoids the waste $c$. Training is REINFORCE with an
annealed entropy bonus; the defector(s) during training draw $D\sim\mathrm{Unif}\{0,
\dots,d_{\max}\}$ per episode so the policy sees a variable, unknown adversary
count.

\textbf{Training schemes.} We compare three ways to provide the adversary, holding
the cooperator architecture fixed:
\emph{(i) single co-training}---cooperators and one co-evolving defector policy
trained together; \emph{(ii) population co-training}---a population of $K{=}4$
co-evolving defectors, one assigned per episode; \emph{(iii) league
(PSRO)}~\cite{lanctot2017psro}---alternate training the cooperators against the
frozen pool of all past best-response defectors and adding a fresh best response to
the pool (Fig.~\ref{fig:arch}b). League training targets the multi-equilibrium
fragility of single co-training: the cooperator becomes robust to the whole
adversary history rather than to one moving target.

\textbf{Baselines and oracle.} Two fixed scripts bracket the problem:
\emph{all-contest} (everyone always claims; $\rho{\approx}1$ but efficiency
$1{-}c$) and \emph{yield} (cooperators always yield; efficiency $1$ but
$\rho{=}N$). The \emph{centralized oracle} allocates by need (worst-off first),
achieving $\rho{=}1$ at efficiency $1$---the unbeatable upper bound a decentralized
policy aims to approach.

\begin{figure*}[t]
\centering
\includegraphics[width=0.96\linewidth]{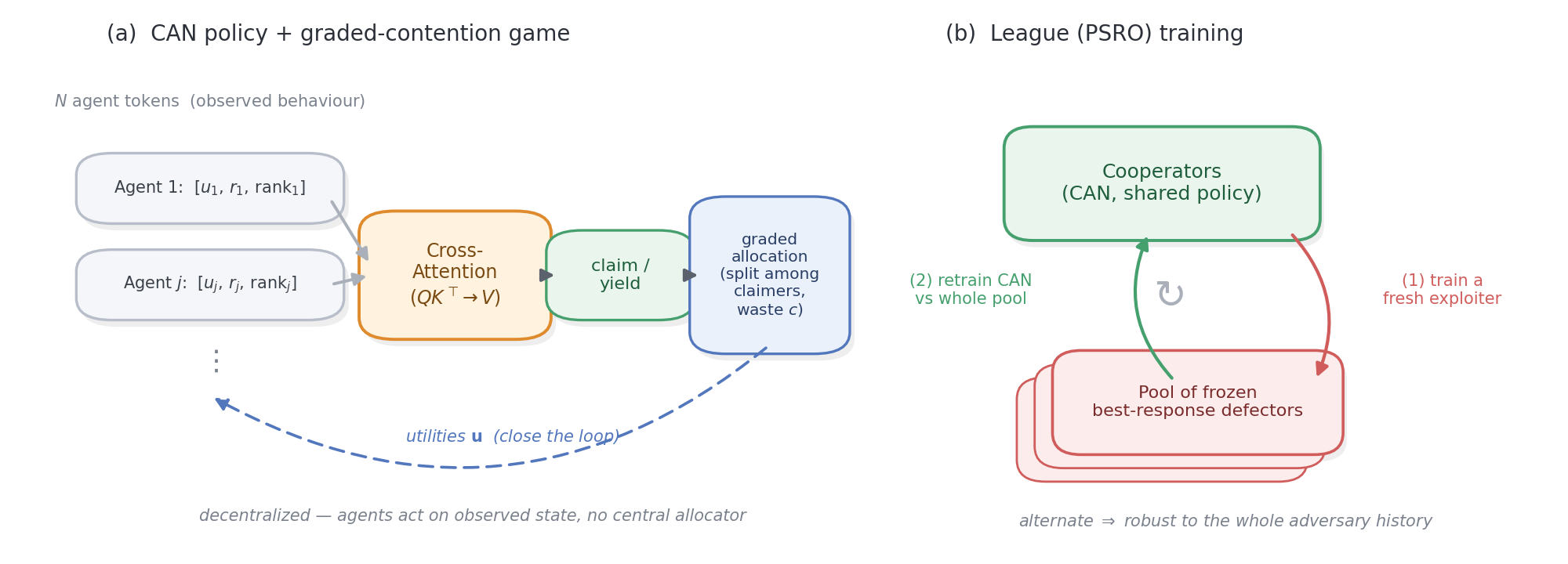}
\caption{Overview. (a) The CAN policy: each agent is a token of observed behaviour
(utilities, claim-rate, who is worst-off); a shared cross-attention block infers
contention and outputs claim/yield; the graded allocator returns utilities, closing
a fully \emph{decentralized} loop. (b) League (PSRO) training: alternate retraining
the cooperators against the frozen pool of past best-response defectors and adding a
fresh exploiter, yielding robustness to the whole adversary history.}
\label{fig:arch}
\end{figure*}

\section{Experiments}
\textbf{Setup.} Unless noted, $N{=}6$, $T{=}100$, $d_{\max}{=}2$,
$c\in\{0.3,0.5,0.7,0.9\}$, $5$ seeds; all policies are shared single-head
cross-attention nets (hidden $64$) trained with REINFORCE/Adam (lr
$3\times10^{-3}$, $B{=}512$ parallel episodes). Every reported policy is audited by
a freshly trained best-response defector (Eq.~\eqref{eq:rho}). Implementation
details are in the appendix.

\subsection{Head-to-head with fair-MARL learners}\label{sec:headtohead}
We first place CAN against the welfare-fair learners that motivate the problem,
\emph{on the same graded-contention game and under the same bounded audit}: a
shared-policy GGF maximizer, a faithful SOTO (Self-/Team-Oriented sub-nets,
annealing $\beta{:}1{\to}0$)~\cite{zimmer2021soto}, and a FEN-style learner whose
agents maximize the fair-efficient reward $\bar u/(c_0{+}|u_i{-}\bar u|)$%
~\cite{jiang2019fen}. All three see the same six behaviour features and are trained
\emph{cooperatively} (no defector at train time, as in their originals); we then
freeze each team and train a best-response defector. This comparison is therefore
\emph{fair-as-published vs.\ robust-by-design}: it isolates whether the fair
\emph{objective alone} confers robustness (it does not). The training axis is
isolated separately: Sec.~\ref{sec:reliability} varies the adversarial scheme on
CAN's own architecture, and Sec.~\ref{sec:hardened} league-trains the fair
objectives themselves under CAN's identical budget---which does not rescue them.

The result is a clean two-axis separation (Fig.~\ref{fig:tradeoff},
Table~\ref{tab:headtohead}). Every fair learner is fair at $D{=}0$ (Jain
${\approx}1.0$), but each collapses onto one of the two scripted extremes and fails
on one axis. \textbf{GGF} learns to \emph{yield}: efficient ($D{=}0$
eff.\ $1.00$) but maximally exploitable ($\rho{=}6{=}N$ at every $c$---a lone
defector takes everything). \textbf{SOTO} learns to \emph{all-contest}: robust
($\rho{=}1.0$) but wasteful---its $D{=}0$ efficiency is exactly $1{-}c$ (down to
$0.10$ at $c{=}0.9$), because its ``be selfish first, then anneal to fair''
schedule settles in a fair-by-symmetry but resource-destroying equilibrium.
\textbf{FEN} is unstable, landing in \emph{either} bad corner depending on seed and
$c$ (at $c{=}0.5$ most seeds learn yield, $\rho{\approx}5.9$, a minority all-contest,
$\rho{=}1.0$, giving a high-variance mean); its fair-efficient reward does not
reliably resolve the tradeoff. \textbf{CAN} is the only method in the good corner:
efficient ($0.98$) \emph{and} robust ($\rho{=}1.2$--$1.5$), clustered beside the
centralized
oracle. The remaining experiments isolate the two properties that put it there.

\begin{table}[t]
\caption{Head-to-head on the graded-contention game ($N{=}6$, 5 seeds). Each fair
learner fails on one axis; CAN attains both. Ranges span $c\in\{0.3,...,0.9\}$.}
\label{tab:headtohead}
\centering
\begin{tabular}{lccl}
\toprule
Method & $D{=}0$ eff. & best-resp.\ $\rho$ & outcome \\
\midrule
GGF (welfare)      & $1.00$        & $6.0$        & efficient, exploitable \\
FEN                & $0.68$–$0.98$ & $3.0$–$6.0$  & unstable, mostly exploitable \\
SOTO               & $1{-}c$       & $1.0$        & robust, wasteful \\
\textbf{CAN (league)} & $\mathbf{0.98}$ & $\mathbf{1.2}$–$\mathbf{1.5}$ & \textbf{both} \\
centralized oracle & $1.00$        & $1.0$        & (upper bound) \\
\bottomrule
\end{tabular}
\end{table}

\begin{figure}[t]
\centering
\includegraphics[width=0.95\linewidth]{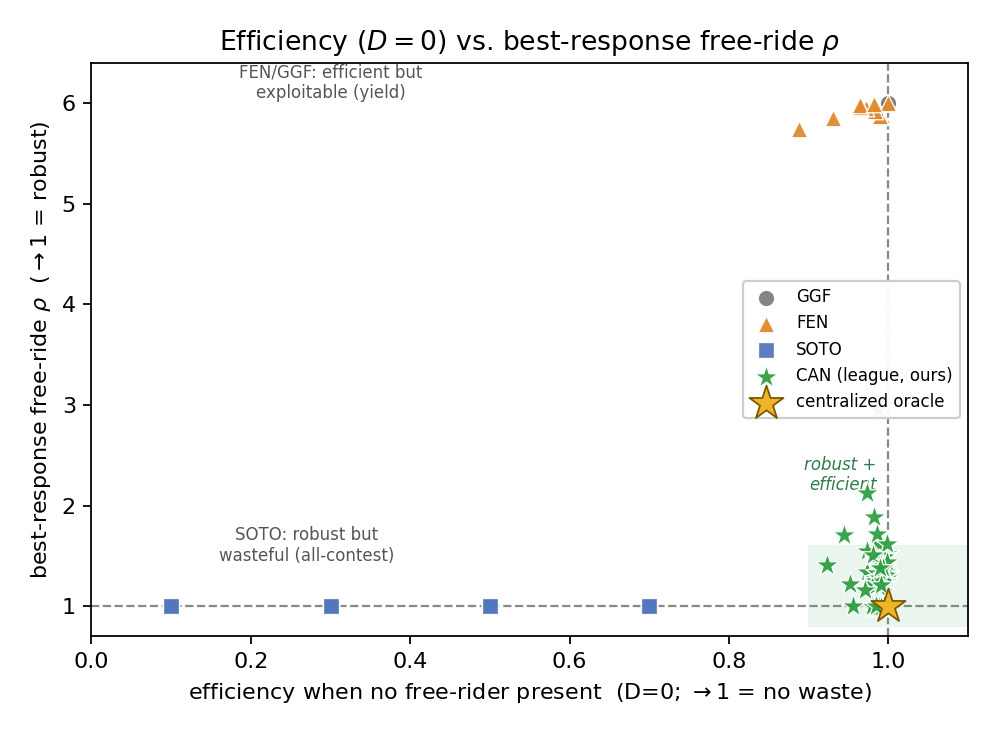}
\caption{Efficiency vs.\ exploitability (each marker is one seed$\times c$). The
welfare-fair learners collapse onto the scripted extremes---GGF/yield (efficient,
$\rho{=}N$) and SOTO/all-contest (robust, eff.\ $1{-}c$)---and FEN scatters between
them. Only CAN occupies the good corner (shaded), beside the centralized oracle.}
\label{fig:tradeoff}
\end{figure}

\subsection{Efficiency when no free-rider is present}
A robust cooperator must not pay for a defense it does not need. Fig.~\ref{fig:eff}
shows the $D{=}0$ efficiency of CAN against the fixed all-contest script. CAN
delivers $0.99$--$1.00$ of the resource at \emph{every} contention level, whereas
all-contest delivers only $1{-}c$ (down to $0.10$ at $c{=}0.9$). Because CAN
turn-takes when it detects no free-rider, it incurs essentially no contention
waste---matching the centralized oracle on efficiency. This is the half of the
problem fixed rules cannot get right.

\begin{figure}[t]
\centering
\includegraphics[width=0.92\linewidth]{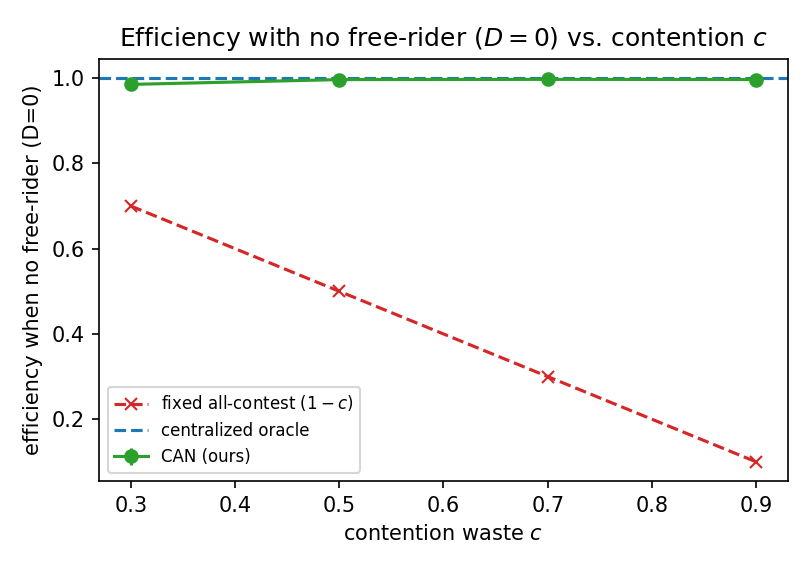}
\caption{Efficiency with no free-rider present ($D{=}0$). CAN turn-takes and wastes
essentially nothing (${\approx}1.0$) at all $c$, while a fixed all-contest policy
pays the full waste $1{-}c$. Mean over 5 seeds.}
\label{fig:eff}
\end{figure}

\subsection{Efficiency \emph{under} a defector: is robustness free?}
$D{=}0$ is the easy case. The axis where a decentralized policy could lose to the
centralized oracle---which delivers $\rho{=}1$ and efficiency $1$
\emph{simultaneously}---is efficiency \emph{with} a free-rider present, since
contesting forces $m\ge2$ and delivers only $1{-}c$ on contested steps; in principle
a policy could buy a low $\rho$ at $c{=}0.9$ by contesting so hard that it wastes as
much as the all-contest equilibrium it is meant to beat. It does not. Under the
best-response defector, league-trained CAN keeps $D{\ge}1$ efficiency at $0.88$,
$0.88$, $0.83$, $0.96$ for $c{=}0.3,0.5,0.7,0.9$ (pooled over $10$ seed-runs)---at the
hardest setting it delivers $0.96$ versus the all-contest equilibrium's
$1{-}c{=}0.10$. CAN contests
\emph{selectively} (just enough, and mostly via turn-taking rather than collisions),
so its robustness is essentially free: it sits near the oracle on \emph{both} axes
even with a defector present, where SOTO reaches $\rho{=}1$ only by destroying $90\%$
of the resource. A residual gap to the oracle remains (efficiency $<1$), as expected
for a decentralized policy, but it is small, not the structural $1{-}c$ collapse.

\subsection{Robustness across the contention range}
Fig.~\ref{fig:robust} and Table~\ref{tab:robust} report best-response
exploitability $\rho$ vs.\ $c$ for CAN trained with vanilla co-training (against a
fixed always-claim adversary) versus league (PSRO) training (against a growing pool
of learned exploiters; to average over run-to-run variance we pool three independent
gen-6 league runs, $15$ seed-runs per $c$). League training keeps $\rho$ low across
the entire range ($1.20$--$1.51$), far below the vanilla co-trained policy
($1.90$--$2.24$) and well below the unprotected $\rho{=}N{=}6$. It is in fact
\emph{most} robust at the highest contention ($\rho{=}1.20$ at $c{=}0.9$, the lowest
point)---consistent with leverage: at higher waste a single contest caps the
free-rider's per-step take at $(1{-}c)/2$, so targeted contention denies it more
effectively (Prop.~1). Vanilla co-training shows the opposite drift, rising to
$2.24$: training against a static adversary does not teach \emph{when} to contest.
Training against a \emph{learned} exploiter is what converts the existence of
leverage (Prop.~1) into robustness. One qualification is established in
Sec.~\ref{sec:hardened}: these $\rho$ values audit with a \emph{trained}
best-response defector, which turns out to be a deterrable adversary; a
\emph{committed} always-claim defector extracts more at high $c$, where the
trained-defector audit is most flattering.

\begin{table}[t]
\caption{Best-response free-ride $\rho$ vs.\ contention $c$ ($N{=}6$, mean with
95\% bootstrap CI; pooled over independent gen-6 runs---league $15$, vanilla $10$
seed-runs per $c$. Lower is better, $1$ = fair, oracle $=1.0$).}
\label{tab:robust}
\centering
\begin{tabular}{lcccc}
\toprule
$c$ & 0.3 & 0.5 & 0.7 & 0.9 \\
\midrule
vanilla co-train & $1.90$ & $2.03$ & $2.06$ & $2.24$ \\
\quad{\footnotesize 95\% CI} & {\footnotesize[1.83,1.96]} & {\footnotesize[1.98,2.09]} & {\footnotesize[1.97,2.15]} & {\footnotesize[2.07,2.42]} \\
CAN (league)     & $1.33$ & $1.36$ & $1.51$ & $\mathbf{1.20}$ \\
\quad{\footnotesize 95\% CI} & {\footnotesize[1.26,1.42]} & {\footnotesize[1.22,1.53]} & {\footnotesize[1.25,1.85]} & {\footnotesize[1.10,1.32]} \\
\bottomrule
\end{tabular}
\end{table}

\begin{figure}[t]
\centering
\includegraphics[width=0.92\linewidth]{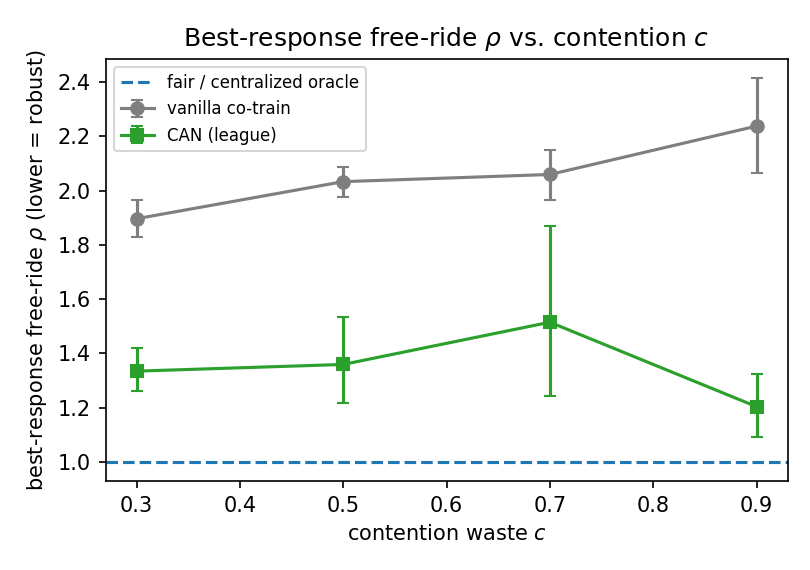}
\caption{Robustness vs.\ contention. League-trained CAN keeps $\rho$ low
($1.20$--$1.51$) and is most robust at the highest contention, while the vanilla
co-trained policy is roughly twice as exploitable and drifts upward. The fair /
centralized-oracle line is $\rho{=}1$. Mean with 95\% bootstrap CI, pooled over
independent gen-6 runs ($15$ league / $10$ vanilla seed-runs).}
\label{fig:robust}
\end{figure}

\subsection{Hardening the fair objectives, and a committed-defector audit}
\label{sec:hardened}
The head-to-head of Sec.~\ref{sec:headtohead} compared fair-as-published against
robust-by-design, leaving open whether the fair objectives could be rescued by
adversarial training alone. We close that cell: GGF, FEN, and SOTO are wrapped in
the \emph{identical} league loop as CAN---same tokens, same $6\times(1200{+}1000)$
update budget, same per-step reward-to-go estimator with entropy annealing, same
$D\sim\mathrm{Unif}\{0..2\}$ training distribution---with only the welfare
function swapped (each reduces to its published objective at $D{=}0$; SOTO's
$\beta$ anneals globally over the league's cooperator updates) and the per-agent
MLP architecture of the originals. Every audit reports the \emph{worse} of two
adversaries: the standard trained best response and a scripted always-claim
defector ($5$ seeds per cell; Table~\ref{tab:hardened},
Fig.~\ref{fig:hardened}).

\textbf{Hardening does not rescue the fair objectives.} FEN+PSRO has exactly two
attractors: all-contest at $c\le0.7$ ($D{=}0$ efficiency exactly $1{-}c$,
$\rho{\approx}1$) and, at $c{=}0.9$, a per-seed bifurcation between yield
(efficient but $\rho{\approx}6{=}N$ \emph{after six league generations}) and
all-contest (efficiency $0.10$). GGF+PSRO and SOTO+PSRO stay efficient
($0.90$--$0.99$) and partially deter the \emph{trained} exploiter
($\rho{\approx}1.4$--$2.3$), but the committed script collects
$\rho{\approx}2$--$4.8$ as $c$ rises. No method reaches the good corner
($D{=}0$ eff.\ $\ge0.95$ and $\rho\le1.6$) beyond $c{=}0.3$. Adversarial training
alone is therefore \emph{not} the mechanism behind CAN's position in
Fig.~\ref{fig:tradeoff}: the welfare-fair objectives are structurally hard to
harden.

\textbf{The dual audit also sharpens CAN's own limit.} Retraining CAN's league
under the dual audit reproduces the trained-best-response $\rho$ of
Table~\ref{tab:robust} within CIs, but the committed always-claim defector
extracts $1.28$ / $1.72$ / $2.29$ / $3.93$ at $c{=}0.3$--$0.9$. The
interpretation of Table~\ref{tab:robust}'s flat $\rho$ therefore changes:
league training produces \emph{deterrence}---the trained exploiter is punished
into hiding, converging to end-game claiming (claim rate
${\approx}0/0/0/0.65$ by episode quarter), the finite-horizon shape of a
deterred adversary---rather than \emph{immunity}. A committed defector ignores
punishment, and sustaining the threat costs the contesters $(1{-}c)/2\to0$, so
robustness to commitment decays exactly as Prop.~1's leverage vanishes. CAN
remains the least exploitable efficient policy at every $c$ under the stricter
score, but the deterrence-vs-commitment gap grows with $c$ for every method,
CAN included. This is the honest empirical match to the theory: decentralized
defense works \emph{as far as leverage reaches}, and no farther.

\begin{table}[t]
\caption{Hardened baselines vs.\ CAN under the dual audit: $D{=}0$ efficiency /
committed-defector exploitability $\max(\rho_{\mathrm{BR}},
\rho_{\mathrm{always}})$, mean over $5$ seeds. No hardened fair objective
reaches the good corner; CAN is best everywhere but decays with $c$.}
\label{tab:hardened}
\centering
\begin{tabular}{lcccc}
\toprule
$c$ & 0.3 & 0.5 & 0.7 & 0.9 \\
\midrule
GGF+PSRO  & 0.94 / 1.60 & 0.98 / 2.04 & 0.99 / 2.88 & 0.99 / 4.76 \\
FEN+PSRO  & 0.70 / 1.04 & 0.50 / 1.02 & 0.37 / 1.78 & 0.63 / 4.07 \\
SOTO+PSRO & 0.90 / 1.55 & 0.98 / 2.02 & 0.97 / 2.50 & 0.99 / 4.26 \\
CAN (league) & \textbf{0.97 / 1.28} & \textbf{0.97 / 1.72} &
\textbf{0.99 / 2.29} & \textbf{1.00 / 3.93} \\
\bottomrule
\end{tabular}
\end{table}

\begin{figure}[t]
\centering
\includegraphics[width=0.95\linewidth]{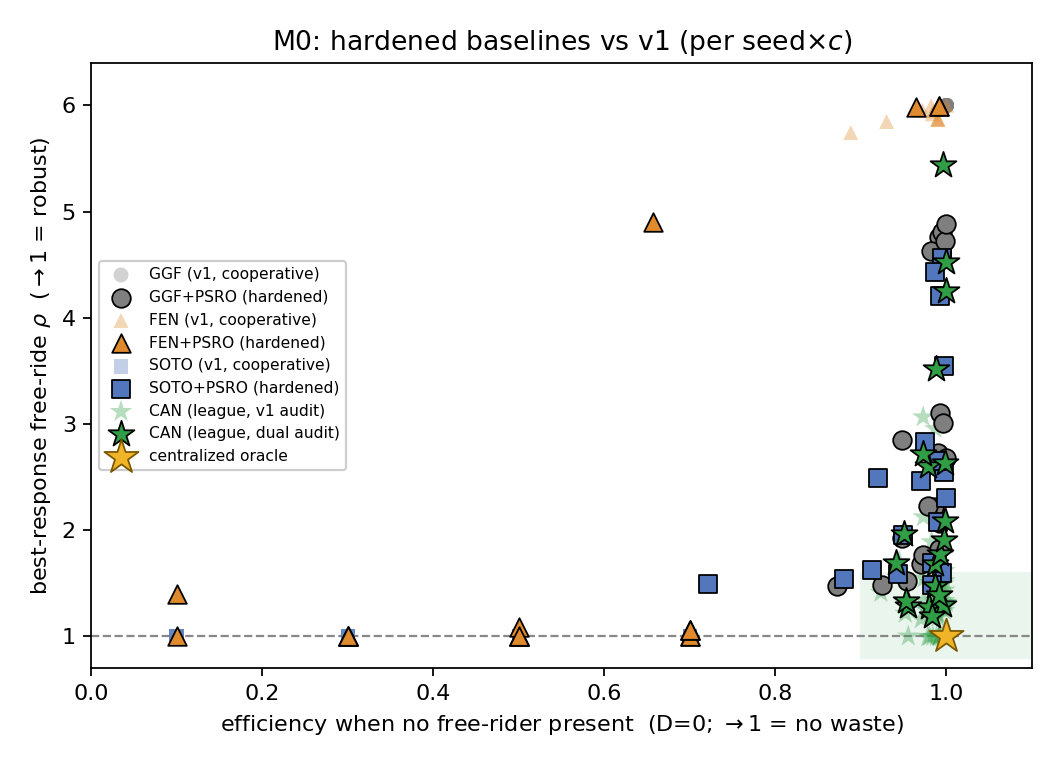}
\caption{Hardened baselines vs.\ v1 on the efficiency--exploitability plane
(per seed$\times c$; hardened points outlined, v1 cooperative training faded).
League-training the fair objectives moves FEN/SOTO toward the wasteful
all-contest corner and leaves GGF efficient but exploitable; no hardened
baseline reaches the good corner. CAN under the dual audit (committed
always-claim included) remains the least exploitable efficient policy but
decays toward high $c$---deterrence, not immunity.}
\label{fig:hardened}
\end{figure}

\subsection{Larger $d_{\max}$ and defector coalitions}
The main results cap the per-episode defector count at $d_{\max}{=}2$, an easy
``0/1/2 grabbers'' inference. We stress this by training CAN with $d_{\max}{=}4$
(up to four of six defecting) and auditing against a best-response \emph{coalition}
of $D$ defectors sharing one policy. CAN remains robust at all $c$ and $D$
(Table~\ref{tab:coalition}): the worst case is a \emph{lone} free-rider
($\rho{=}1.13$--$1.60$, comparable to the $d_{\max}{=}2$ results), and larger
coalitions are \emph{easier}---$\rho{\to}1.0$ for $D{\ge}2$---because with more
agents defecting there is less cooperative surplus to appropriate. The harder
inference and coordinated free-riders thus do not break the policy.

\begin{table}[t]
\caption{Larger-$d_{\max}$ stress: best-response $\rho$ for CAN trained at
$d_{\max}{=}4$, audited against a $D$-defector coalition ($N{=}6$, 5 seeds).}
\label{tab:coalition}
\centering
\begin{tabular}{lccc}
\toprule
$c$ & $D{=}1$ & $D{=}2$ & $D{=}3$ \\
\midrule
0.3 & 1.22 & 1.11 & 1.07 \\
0.5 & 1.13 & 1.00 & 1.00 \\
0.7 & 1.51 & 1.00 & 1.00 \\
0.9 & 1.60 & 1.05 & 1.00 \\
\bottomrule
\end{tabular}
\end{table}

\subsection{Choice of adversarial training scheme}\label{sec:reliability}
Does the result depend on the \emph{league} specifically? Fig.~\ref{fig:rel}
compares three adversarial schemes---single co-training (one co-evolving defector),
population co-training ($K{=}4$ defectors), and the league---at the two lower
contention levels, per seed. The schemes are statistically indistinguishable: at
$c{=}0.3$ all three average $\rho{\approx}1.48$, and at $c{=}0.5$ they span
$1.18$--$1.46$ with overlapping CIs. What carries the result is \emph{that} the
cooperator is trained against a \emph{learned}, adapting exploiter at all: every
adversarial scheme lands near $\rho{=}1.2$--$1.5$, versus $1.8$--$2.2$ for the
static-adversary vanilla baseline (Table~\ref{tab:robust}). We adopt the league for
its robustness to the whole adversary history, but the headline robustness is not an
artifact of that choice.

\begin{figure}[t]
\centering
\includegraphics[width=0.98\linewidth]{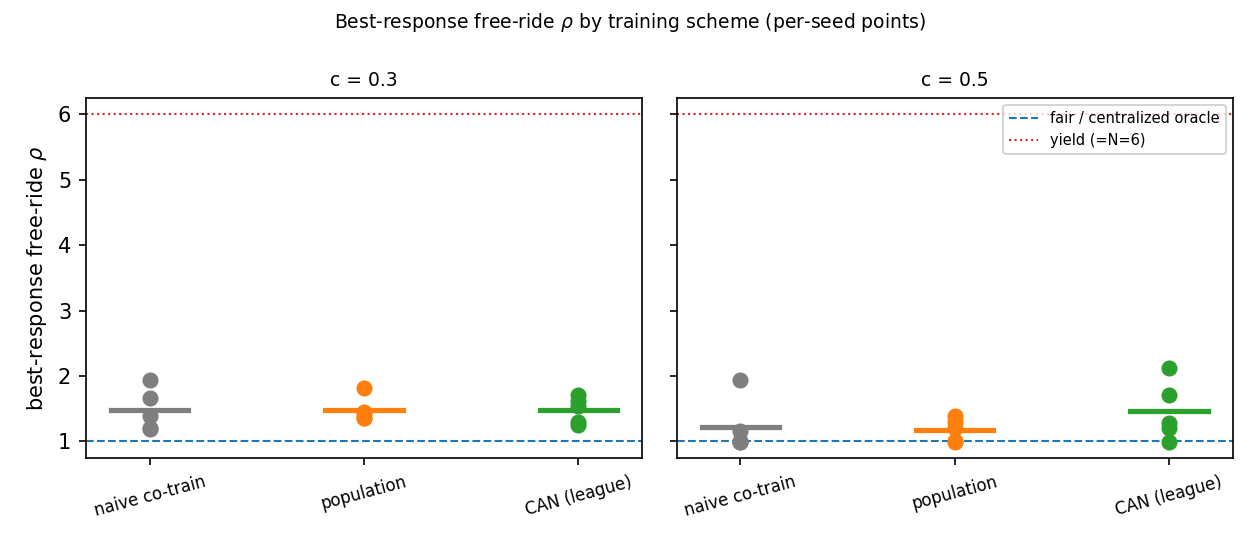}
\caption{Per-seed best-response $\rho$ for three adversarial training schemes (5
seeds). The schemes are indistinguishable (all $\rho{\approx}1.2$--$1.5$); what
matters is training against a \emph{learned} exploiter rather than a static one. The
dotted line marks the all-yield ceiling ($\rho{=}N$).}
\label{fig:rel}
\end{figure}

\subsection{What is load-bearing? Architecture ablation}\label{sec:ablation}
Is the win the attention, or just behaviour-conditioning with \emph{some} adaptive
aggregation? We test this by replacing the cross-attention block with three
alternative aggregators of the \emph{same} behaviour tokens, under \emph{identical}
league training and the same best-response audit: a \emph{mean-pool} policy (each
agent conditions on its token and the mean of all tokens), a \emph{deep-sets} policy
(mean-pool of a learned per-agent embedding), and a bidirectional \emph{GRU} over the
agent tokens. Table~\ref{tab:ablation} reports best-response $\rho$ (mean over $c$, 5
seeds).

The clear conclusion is that \emph{aggregation matters but pooling is not enough}:
the mean-pool and deep-sets policies are markedly more exploitable than CAN
($\rho{=}1.78,1.75$ vs.\ $1.37$), so the natural objection ``a
permutation-equivariant pool should already match attention'' does not hold. The
gap to the strongest alternative is narrower and we do not overstate it: the
bidirectional GRU is competitive on the mean ($1.56$ vs.\ $1.37$) and even better at
low $c$, so we credit \emph{behaviour-conditioning plus adaptive aggregation}, not
attention per se. Where cross-attention earns its place is \emph{stability}: at the
hardest contention ($c{=}0.9$) CAN holds $\rho{=}1.16{\pm}0.20$ while the GRU is
erratic ($1.99{\pm}1.50$), so attention is the most stable instance of the family
rather than categorically superior. These comparisons are at the deployed training
budget; we found that under a
\emph{lighter} budget the CAN--GRU ranking reverses, underscoring that architecture
ablations must use the budget at which the method is actually trained.

\begin{table}[t]
\caption{Architecture ablation under identical league training: best-response
$\rho$ (mean over $c$, 5 seeds; lower is better). Pooling policies ($\dagger$, run
at a lighter budget) are more exploitable than CAN \emph{even when CAN is held to
that same lighter budget} ($\rho{=}1.57$); at full budget CAN improves to $1.37$
and also beats the bi-GRU, chiefly by staying stable at the hardest contention.
CAN's $\rho$ in this ablation ($1.37$) is one of the three independent gen-6 league
runs pooled in Table~\ref{tab:robust} (pooled mean $1.35$); all audit with a fresh
best-response defector, and the spread reflects run-to-run (GPU) variance.}
\label{tab:ablation}
\centering
\begin{tabular}{lcc}
\toprule
Policy (aggregator) & best-resp.\ $\rho$ & $\rho$ at $c{=}0.9$ \\
\midrule
mean-pool$^\dagger$            & 1.78 & 2.05 \\
deep-sets$^\dagger$            & 1.75 & 2.06 \\
bidirectional GRU              & 1.56 & $1.99{\pm}1.50$ \\
\textbf{cross-attention (CAN)} & \textbf{1.37} & $\mathbf{1.16{\pm}0.20}$ \\
\bottomrule
\end{tabular}
\end{table}

\begin{figure}[t]
\centering
\includegraphics[width=0.92\linewidth]{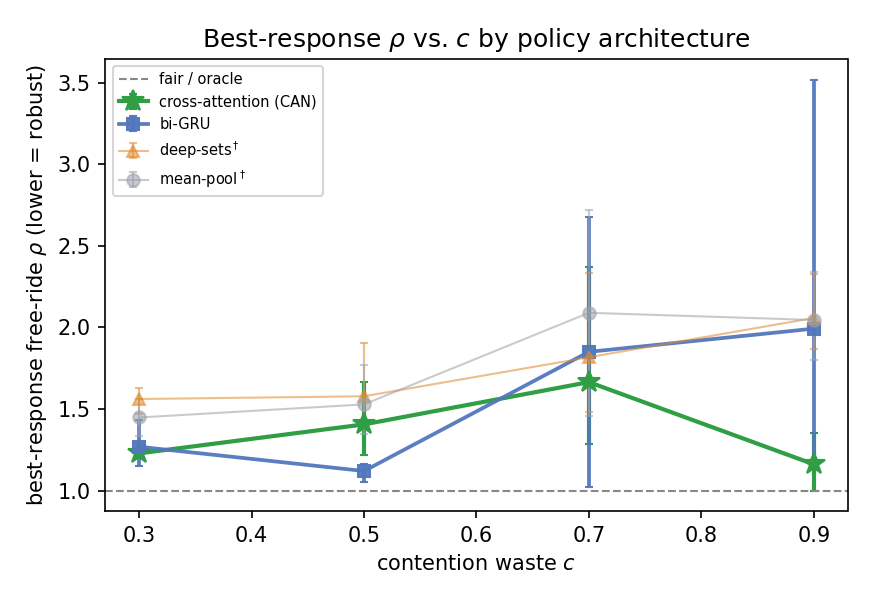}
\caption{Architecture ablation: best-response $\rho$ vs.\ $c$ by aggregator (5
seeds, 95\% bootstrap CI). Cross-attention is the only policy that stays low
\emph{and} becomes most robust at the hardest contention ($c{=}0.9$); the bi-GRU is
competitive at low $c$ but erratic at high $c$, and the pooling baselines
($\dagger$, lighter budget) are uniformly worse.}
\label{fig:ablation}
\end{figure}

\subsection{Zero-shot transfer across team size (a limitation)}
Because CAN is permutation-equivariant it can be evaluated at team sizes it never
trained on. Fig.~\ref{fig:transfer} applies the $N{=}6$ policy to teams of $12$ and
$24$. Transfer holds well at low contention---at $c{=}0.3$ a $4\times$ larger team
is still close to fair ($\rho{=}1.85$ at $N{=}24$)---but degrades at high
contention, reaching $\rho{=}8.8$ at $N{=}24$, $c{=}0.9$. The bounded metric keeps
these readable (all curves stay below the all-yield ceiling $\rho{=}N$), and the
ordering is clean: difficulty grows with both team size and waste. We report this as
a genuine fragility---the contention-inference learned at $N{=}6$ does not fully
generalize to crowds at high waste---rather than smoothing it over. The obvious
remedy does not work: training on a size curriculum ($N\in\{4,6,8,12\}$, evaluating
at the held-out $N{=}24$) helps only at $c{=}0.3$ and \emph{worsens} mid-range
transfer (it dilutes the per-size budget), leaving the high-contention gap unchanged
($\rho{=}4.5$ at $c{=}0.9$, vs.\ $4.5$ for $N{=}6$-only).

\begin{figure}[t]
\centering
\includegraphics[width=0.92\linewidth]{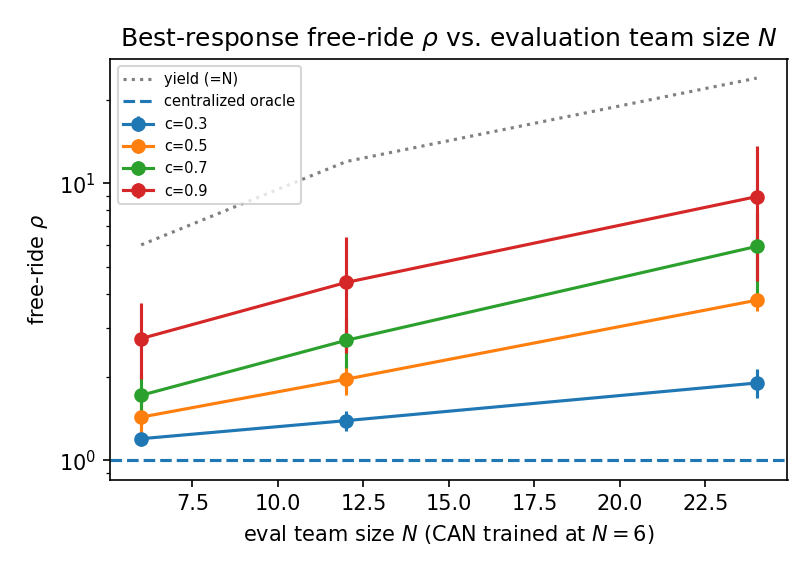}
\caption{Zero-shot transfer of the $N{=}6$ policy to larger teams (log scale). Robust
at low contention; degrades with team size at high contention. All curves remain
below the all-yield ceiling $\rho{=}N$; the oracle line is $\rho{=}1$.}
\label{fig:transfer}
\end{figure}

Two follow-up experiments locate the failure precisely, and it is \emph{not} team
size per se (Fig.~\ref{fig:transfer_count}). First, re-parameterizing the tokens
to be $N$-invariant in scale
(utility \emph{shares} and z-scored deviations in place of $T$-normalized levels,
plus explicit claimed-fraction channels) leaves the high-$c$ gap unchanged
($\rho{=}9.2$ at $N{=}24$, $c{=}0.9$), and the trained policy's attention stays
proportionally sharp at every team size---so neither input scale drift nor
attention dilution is the cause. Second, evaluating at a constant defector
\emph{fraction} ($D{=}N/6$ rather than a lone defector), the unmodified $N{=}6$
policy transfers essentially perfectly: $\rho{=}2.5\to1.5$ from $N{=}6$ to
$N{=}48$ at $c{=}0.9$, flat or improving at every $c$ (dashed curves in
Fig.~\ref{fig:transfer_count}). What degrades with $N$ is
the \emph{detection of a single defector}: a lone free-rider is a vanishing
fraction of every aggregate the policy can form (its weight under any convex
pooling scales ${\sim}1/N$), so the onset of contestation arrives later as $N$
grows---at $N{=}48$, $c{=}0.9$ cooperators are nearly silent for the first
${\sim}30$ steps, during which the sole-claiming defector banks full, waste-free
units. The free-ride against a lone defector is a \emph{count} quantity, not a
fraction quantity, and no mean-field-style parameterization preserves it.
Breaking the $1/N$ scaling is partially possible
(Fig.~\ref{fig:transfer_fix}): adding a count-preserving
max-pooled suspicion branch to the aggregation, fed by an instantaneous
claimed-last-step channel, restores near-immediate detection at $N{=}48$ and
roughly halves the lone-defector $\rho$ twice over ($N{=}24$, $c{=}0.9$:
$8.9\to4.6$; $N{=}48$: $20.2\to6.9$; the channel is useless without the max
statistic and vice versa at high $c$). What remains is a \emph{structural
floor}: with detection latency $L$ steps, the defector banks $L$ full units
while every later contested step delivers only $1{-}c$, so
$\rho\approx N L/(L+(T{-}L)(1{-}c))$ stays high as $c\to1$ even under
near-perfect detection---the price of the detection window, not an inference
failure.

\begin{figure*}[t]
\centering
\includegraphics[width=0.99\linewidth]{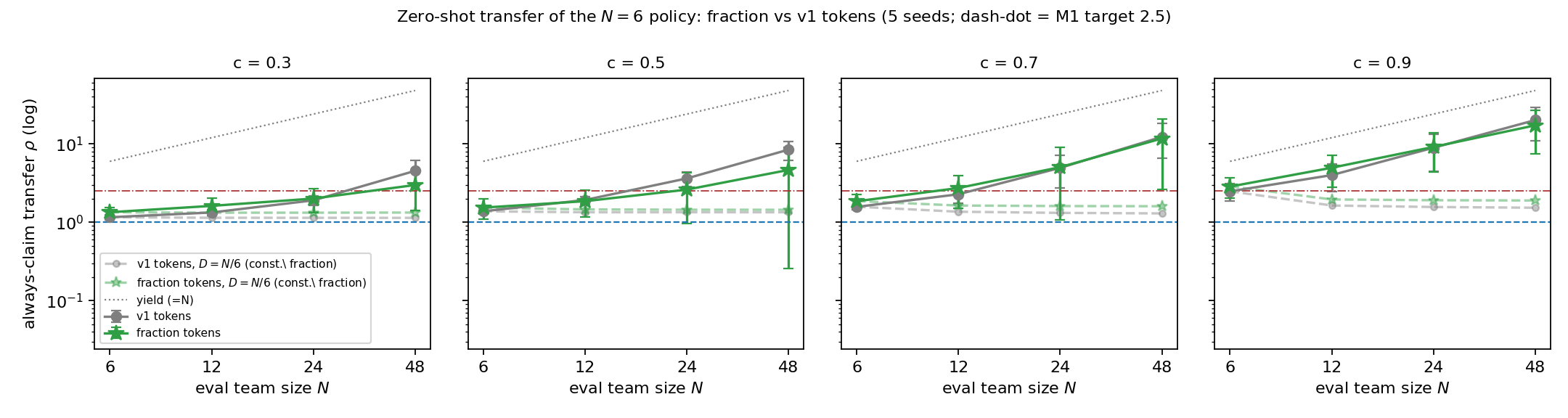}
\caption{The transfer fragility is the \emph{lone} defector, not team size.
Solid: lone-defector ($D{=}1$) zero-shot transfer of the $N{=}6$ policy, v1 vs
$N$-invariant fraction tokens---both degrade toward the all-yield line as $N$
grows at high $c$. Dashed: the same policies at constant defector
\emph{fraction} ($D{=}N/6$) transfer essentially perfectly at every $c$.
Log--log; dash-dot line marks $\rho{=}2.5$.}
\label{fig:transfer_count}
\end{figure*}

\begin{figure*}[t]
\centering
\includegraphics[width=0.99\linewidth]{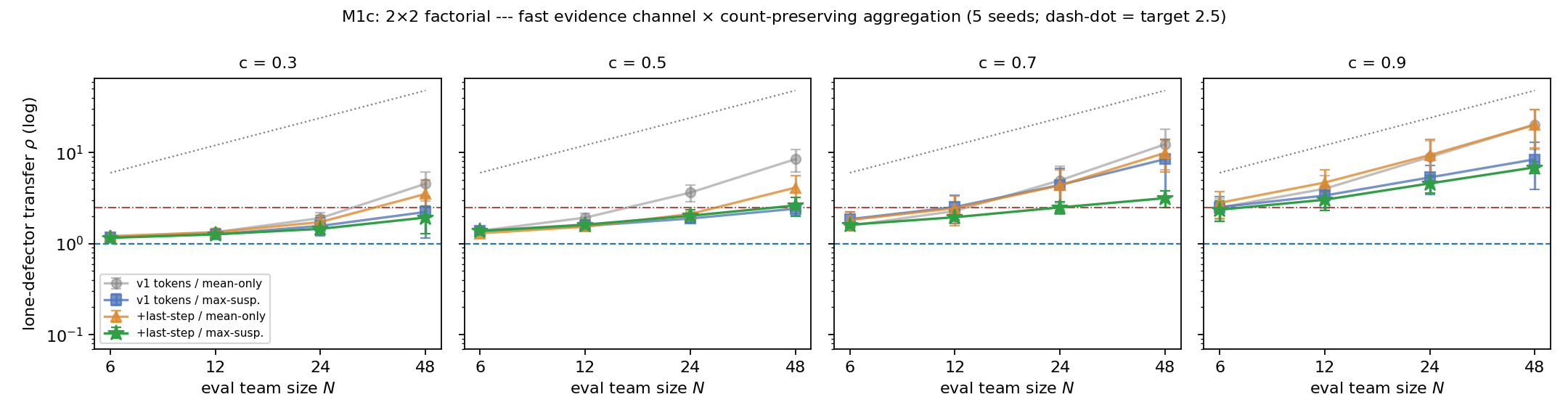}
\caption{Breaking the $1/N$ evidence scaling ($2{\times}2$ factorial). The
instantaneous claimed-last-step channel alone ($\triangle$) matches the v1
control---fast evidence is useless under convex (mean-style) pooling; the
count-preserving max-suspicion branch alone ($\square$) halves the
lone-defector gap; their combination ($\star$) halves it again and meets the
$\rho{<}2.5$ target through $c{=}0.7$. The $c{=}0.9$ residual is the
structural floor $\rho\approx NL/(L{+}(T{-}L)(1{-}c))$.}
\label{fig:transfer_fix}
\end{figure*}

\subsection{Generality across environments}\label{sec:generality}
The experiments so far use the single-resource graded game. To test whether CAN's
behaviour-conditioned contention inference is specific to that game, we replicate the
head-to-head on two further leverage-preserving environments that pose
\emph{different} inference problems (Fig.~\ref{fig:tradeoff_envs},
Table~\ref{tab:generality}). \emph{(i) Congestion}: $M{=}3$ parallel servers; each
agent claims one or yields; equal-split per server; empty servers route to the
neediest. The cooperators must infer \emph{which} server each free-rider targets and
contest just that one. \emph{(ii) Stakes}: a single resource whose value $v_t$ is
random (lumpy: usually small, occasionally a jackpot, $\mathbb{E}[v]{=}1$). The
cooperators must contest \emph{selectively}, on the high-stakes steps where a
free-rider's surplus concentrates. Both keep equal-split contention, so Prop.~1
leverage holds; CAN and the three welfare-fair learners are trained and audited
exactly as before.

\textbf{The efficiency axis is what carries over.} On exploitability alone the games
differ in difficulty (the congestion exploit is structurally capped at
$\rho{=}N/M{=}2$; stakes restores the $\rho{=}N$ ceiling). The one property that is
consistent is efficiency: \textbf{CAN keeps $D{=}0$ efficiency $0.94$--$1.00$ in all
three games}, so it never buys robustness by wasting the resource---unlike the only
baseline that attains low $\rho$, SOTO, whose all-contesting collapses efficiency to
$0.40$--$0.47$ (and to $0.10$ at stakes $c{=}0.9$). GGF and FEN are Pareto-dominated
by CAN in every game. Robustness, however, does \emph{not} carry over uniformly: CAN
is both efficient and robust on base and congestion, but on stakes it is only
\emph{partially} robust---robust up to moderate contention and exploited at
$c{=}0.9$---so we claim a clear scope, not blanket generality.

\textbf{An honest gradient: leverage strength governs success.} CAN is not uniformly
robust; its exploitability tracks how much leverage Prop.~1 actually delivers. With
\emph{strong} leverage (base: a contest caps the free-rider at $(1{-}c)/2$ for any
$c$) CAN is robust at every $c$ and \emph{most} robust at the highest contention
($\rho{=}1.20$ at $c{=}0.9$). With \emph{structurally capped} leverage (congestion)
it is robust but cannot beat the $\rho{=}2$ ceiling by much at high $c$. With
\emph{weakening} leverage (stakes: contesting a jackpot at $c{=}0.9$ returns only
$0.05\,v$) it degrades to $\rho{=}4.05$ at $c{=}0.9$---while still staying efficient
($0.97$) where the robust alternative (SOTO) delivers only $0.10$. And with
\emph{absent} leverage---a rich-get-richer (Matthew) rule where the \emph{richest}
claimer wins, so a worst-off contester captures nothing and Prop.~1 fails by
construction---CAN is exploited ($\rho{\approx}5.5$), marking the boundary of the
approach. The method works as far as the leverage it is built on, and no further.

\begin{table}[t]
\caption{Generality: ($D{=}0$ efficiency, best-response $\rho$) per method, mean over
$c$ (5 seeds). CAN is efficient ($\ge0.94$) in every environment and the least
exploitable efficient policy; SOTO's robustness costs catastrophic efficiency.
$\rho_{\max}$ is the yield ceiling.}
\label{tab:generality}
\centering
\begin{tabular}{lcccccc}
\toprule
& \multicolumn{2}{c}{base ($\rho_{\max}{=}6$)} & \multicolumn{2}{c}{congestion ($2$)}
& \multicolumn{2}{c}{stakes ($6$)} \\
\cmidrule(lr){2-3}\cmidrule(lr){4-5}\cmidrule(lr){6-7}
Method & eff & $\rho$ & eff & $\rho$ & eff & $\rho$ \\
\midrule
GGF  & 1.00 & 6.00 & 1.00 & 2.01 & 0.98 & 5.75 \\
FEN  & 0.81 & 4.20 & 0.96 & 2.05 & 0.70 & 3.27 \\
SOTO & 0.40 & 1.00 & 1.00 & 2.01 & 0.47 & 1.64 \\
\textbf{CAN} & \textbf{0.98} & \textbf{1.35} & \textbf{1.00} & \textbf{1.60}
 & \textbf{0.94} & \textbf{2.32} \\
\bottomrule
\end{tabular}
\end{table}

\begin{figure*}[t]
\centering
\includegraphics[width=0.98\linewidth]{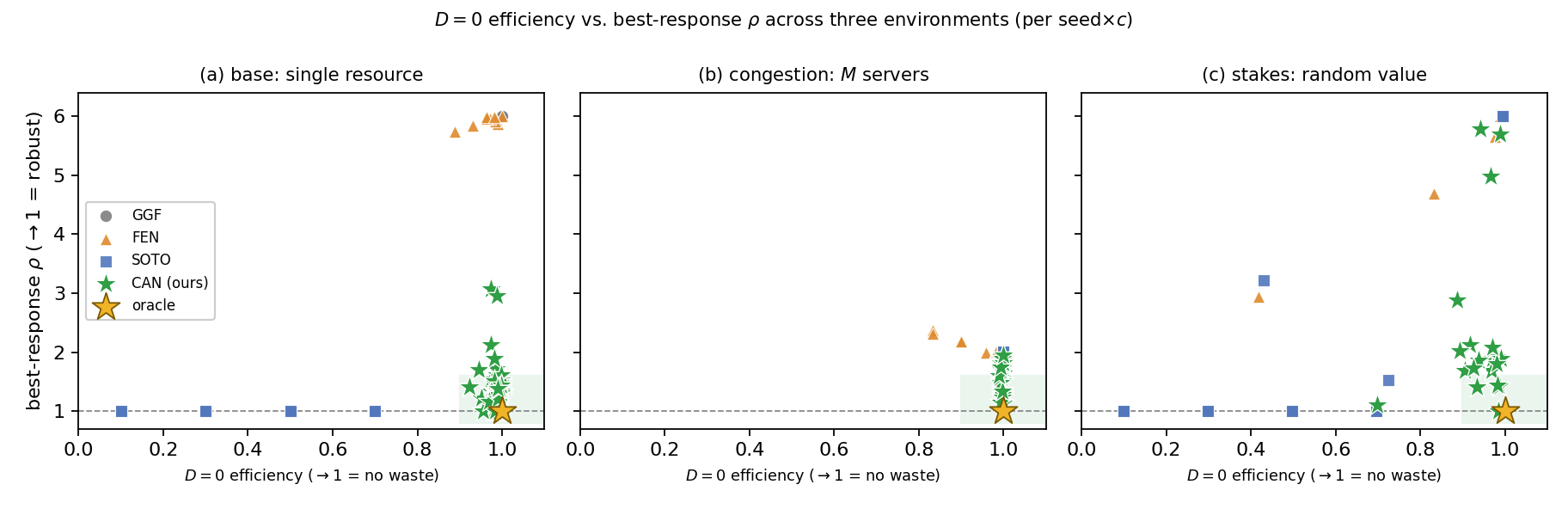}
\caption{$D{=}0$ efficiency vs.\ best-response $\rho$ on three structurally distinct
leverage-preserving games (each marker is one seed$\times c$). In every game CAN
(green) stays on the efficient (right) side while remaining the least exploitable
efficient policy; the only robust baseline, SOTO (blue), reaches low $\rho$ only by
moving left into the wasteful region. GGF/FEN are efficient but exploitable
(top-right). Shaded = the efficient-and-robust corner; gold star = centralized
oracle.}
\label{fig:tradeoff_envs}
\end{figure*}

\section{Discussion and Limitations}
\textbf{What this shows.} The decentralization gap left open by prior
work~\cite{savci2026exploit} is not fundamental: on the graded game a decentralized
cross-attention policy trained against an adversarial league recovers most of the
centralized oracle's incentive-compatibility against adaptive exploiters
($\rho{\approx}1.2$--$1.5$; committed defectors are deterred only as far as
leverage reaches, Sec.~\ref{sec:hardened}) \emph{and}
its efficiency (${\approx}1.0$ when no free-rider is present), with no central
allocator. The all-or-nothing futility result is the boundary case $c{=}1$. The
effect is not automatic for every graded rule, though: it requires that contesting a
free-rider actually return a positive share to the contester (Prop.~1), which the
stakes and Matthew games show can be weak or absent.

\textbf{Where it is fragile.} (i) Zero-shot transfer against a \emph{lone}
defector degrades with team size at high contention (Fig.~\ref{fig:transfer});
neither a size curriculum nor $N$-invariant tokens fix it, because the cause is
detection latency---a single defector is a vanishing fraction of every aggregate,
so evidence scales ${\sim}1/N$. Transfer at constant defector fraction already
works; crowd-scale detection of a lone free-rider at high waste is the open
problem. (ii) The single-adversary
scheme
is seed-sensitive at low contention; league training is needed for consistency, at
extra training cost. (iii) Robustness is bounded by leverage: on the stakes game it
degrades at high contention (where contesting a jackpot returns almost nothing), and
under a winner-take-all (Matthew) rule it disappears entirely---CAN is robust only
as far as Prop.~1's leverage extends. (iv) What league training buys is
\emph{deterrence} of adaptive exploiters, not immunity: a \emph{committed}
always-claim defector, which no amount of punishment dissuades, extracts more as
$c\to1$ (Sec.~\ref{sec:hardened})---for every method, ours included. In
boundary experiments at $c{=}1$ the obstacle is \emph{credibility}, not
detection: scripted trigger policies over the same public signals (grim /
tit-for-tat on observed claim rates) hold best-response and patient learning
defectors at exactly $\rho{=}1.0$ at zero efficiency cost, while league-learned
policies, whose punishment is probabilistic rather than absorbing, leave
$\rho{\approx}1.8$. This does not mean one should simply deploy the trigger:
the scripts are hand-designed around a known threat model (a claim-rate flag
with a tuned threshold) and are correct only at the boundary---where
punishment is free because the contested unit is destroyed anyway---whereas at
$c<1$ permanent punishment burns surplus that proportional contesting would
recover, which is exactly the regime the learned policy handles. The open
problem is \emph{learned, transferable} commitment: a single policy that
acquires the trigger's credibility where leverage is absent without inheriting
its rigidity where leverage exists. That is the subject of companion work;
code for those experiments is included in the repository. (v) Our games
are controlled abstractions with
a single shared cooperator policy; spatial/continuous environments and
repeated-interaction benchmarks are still open
(coordinated defector coalitions, by contrast, we do test, and they are no harder
than a lone free-rider). (vi) The oracle remains an upper bound CAN approaches but
does not reach.

\textbf{Takeaway.} Fair-MARL methods are usually reported on the fairness they
achieve under full cooperation. We argue they should also be reported on the
fairness they \emph{retain} under self-interest---and that, contrary to the
all-or-nothing intuition, a decentralized policy can retain much of it whenever the
contention rule leaves a contester a real share to capture.

\section{Conclusion}
We gave a decentralized account of robust fairness in cooperative MARL. Graded
contention makes contesting a free-rider strictly worthwhile for the worst-off
(Prop.~1), but exploiting this requires inferring an unknown, variable number of
free-riders and responding proportionally. CAN---a permutation-equivariant
cross-attention policy over observed behaviour, trained against an adversarial
league---does so, holding trained-best-response exploitability low across the
contention range while wasting essentially nothing when no free-rider is
present, thereby approaching a centralized need-based oracle without a central
allocator; against a committed defector this protection is a deterrence effect
that holds in proportion to the remaining leverage. We are explicit about scope rather than claiming blanket generality:
across additional games CAN stays efficient and dominates the welfare-fair learners,
but its robustness holds in proportion to the contest leverage---strong on a
multi-server game, only partial where the leverage weakens (high-value steps at high
contention), and absent under a winner-take-all rule. We hope this reframes
decentralized robust fairness as a tractable, measurable target---bounded by a clear
condition---rather than an impossibility.

\section*{Code Availability}
All code (the graded-contention game, CAN, the architecture and training-scheme
ablations, and all experiments) and the \LaTeX{} source are available at
\url{https://github.com/highcansavci/can-fair-marl}.

\section*{Use of AI-Assisted Technologies}
The author used Claude (Anthropic), including Claude Code, to assist with the
experimental code, manuscript review and revision, and formatting. All AI-assisted
content was verified by the author, who is solely responsible for the methodology,
results, and conclusions of this paper.

\appendices
\section{Environment Details}
All environments share a skeleton: $N{=}6$ agents act for $T{=}100$ steps; at each
step divisible resource is available; each agent chooses a self-interested act
(\textsc{Claim}) or to \textsc{Yield}; accumulated utilities $\mathbf{u}$ determine
fairness. They differ only in the \emph{allocation rule} mapping the joint action to
utility increments, and (for stakes/congestion) in the action space and one extra
observation. Let $m$ be the number of claimers of a resource and $c\in[0,1]$ the
contention waste.

\textbf{Base (single-resource graded contention).} One resource per step, binary
action. Each claimer receives
\begin{equation*}
g(m)=\begin{cases}
1 & m{=}1,\\
(1{-}c)/m & m\ge 2,\\
\text{(unit to worst-off } \arg\min_i u_i) & m{=}0,
\end{cases}
\end{equation*}
so a sole claimer wins the unit, $m\ge2$ claimers split a discounted $1{-}c$, and an
unclaimed unit flows losslessly to the neediest agent. $c{=}1$ recovers the
all-or-nothing model.

\textbf{Congestion (multi-server).} $M{=}3$ parallel unit resources (``servers'')
per step; action space $\{\textsc{Yield},\,\text{claim }1,\dots,\text{claim }M\}$
($M{+}1$ actions). Each server is allocated independently by the base rule ($k$
claimers $\!\to\!1$ if $k{=}1$, $(1{-}c)/k$ if $k\ge2$); the $E$ empty servers route
to the $E$ lowest-utility agents (one unit each, no waste). A lone free-rider can
monopolise only one of $M$ servers, so the exploit ceiling is $\rho{=}N/M{=}2$. Each
token is augmented with the agent's per-server running claim-rates (so the policy
can infer \emph{which} server is being grabbed): $5{+}M$ features, $(M{+}1)$-way
head.

\textbf{Stakes (stochastic value).} One resource per step with random value $v_t$
drawn i.i.d.\ as $v_t{=}3.25$ w.p.\ $0.25$ and $v_t{=}0.25$ otherwise (lumpy
jackpots, $\mathbb{E}[v_t]{=}1$, so efficiency normalises as in base). The base rule
scales by $v_t$: a sole claimer gets $v_t$, $m\ge2$ claimers get $(1{-}c)v_t/m$ each,
an unclaimed unit gives $v_t$ to the worst-off. The current $v_t$ is appended to each
token ($7$ features), so the policy can contest selectively on high-value steps.

\textbf{Matthew (rich-get-richer; boundary).} One resource, binary action, but among
claimers the \emph{richest} wins: $m{=}1$ gives the sole claimer $1$; $m\ge2$ gives
the richest claimer $(1{-}c)$ and the rest $0$ (ties broken uniformly at random);
$m{=}0$ routes to the worst-off. A worst-off cooperator that contests a richer
free-rider receives $0$, so Prop.~1 leverage vanishes; we use this environment only
to mark the boundary of the method.

\textbf{Threat model and metric (all environments).} Each episode, $D$ agents are
\emph{defectors} with $D\sim\mathrm{Unif}\{0,\dots,d_{\max}\}$ ($d_{\max}{=}2$ unless
noted; the stress study uses $d_{\max}{=}4$); defectors always take the
self-interested act and the $N{-}D$ cooperators share one policy. To audit a frozen
cooperator policy we train a best-response defector (or a $D$-agent coalition sharing
one policy) to maximise the defector group's utility, and report the bounded
free-ride factor
\begin{equation*}
\rho=\frac{N\sum_{i\in\mathcal{D}}u_i}{\,n_{\mathrm{def}}\sum_j u_j\,}\in[0,\,N/n_{\mathrm{def}}],
\end{equation*}
$\rho{=}1$ exactly fair, $\rho{=}N$ a lone defector taking everything. Cooperators
maximise the per-step increment of
$W_{\mathrm{coop}}{=}\mathrm{mean}_{i\in\mathcal{C}}(u_i){-}\mathrm{std}_{i\in\mathcal{C}}(u_i)$.

\section{Implementation Details}
CAN and all adversaries are shared single-head attention blocks (the inter-agent
cross-attention of Sec.~IV, computed as one self-attention pass over the
$N{\times}6$ token matrix; hidden width $64$): $Q,K,V$ linear maps,
scaled-dot-product attention, a $\tanh$ layer on $[\,x_i\,\|\,\mathrm{ctx}_i\,]$,
and a $2$-logit head.
All policies train with REINFORCE and Adam (lr $3\times10^{-3}$), $B{=}512$
parallel episodes, episodes of $T{=}100$ steps, with the entropy coefficient
annealed from $0.05$ to $0.003$. Per-episode defector count is
$D\sim\mathrm{Unif}\{0,1,2\}$. Vanilla/population co-training run $3000$ updates;
league training runs $6$ generations of $1200$ cooperator updates against the pool
and $1000$ updates per fresh best response. Every audit trains a best-response
defector for $1500$ updates against the frozen cooperator and reports $\rho$ from
Eq.~\eqref{eq:rho}; the dual audit of Sec.~\ref{sec:hardened} additionally rolls
out a scripted always-claim defector at a random index and reports the worse of
the two $\rho$ values per seed. The hardened baselines of Sec.~\ref{sec:hardened}
reuse this league setup verbatim with the cooperator welfare replaced by the
respective published objective masked to cooperators (GGF over cooperator
utilities; FEN's fair-efficient reward with the mean over cooperators; SOTO's
self/team sub-nets with $\beta{:}1{\to}0$ annealed globally over the
$6\times1200$ cooperator updates) on the per-agent MLP (hidden $64$) of the
cooperative baselines. The cooperator welfare is the masked
$\mathrm{mean}{-}\mathrm{std}$ over cooperators; advantages are discounted
($\gamma{=}0.99$) per-step welfare increments with a batch-mean (or learned value)
baseline. The architecture ablation reuses this league setup with the
cross-attention block replaced by a mean-pool, deep-sets, or bi-GRU aggregator; the
mixed-$N$ curriculum cycles the team size over $\{4,6,8,12\}$ per update with the
same total step budget. Unless noted, all numbers average $5$ seeds and error bars
are $95\%$ bootstrap confidence intervals. League training is mildly sensitive to GPU
non-determinism, which compounds over generations; the headline league/vanilla
robustness (Table~\ref{tab:robust}, Fig.~\ref{fig:robust}) therefore pools several
independent gen-6 runs ($15$ and $10$ seed-runs) so the CIs absorb run-to-run as well
as seed variance.

\bibliographystyle{IEEEtran}
\bibliography{refs}

\end{document}